\documentclass[12pt]{article}
\usepackage{amsthm,natbib}
\usepackage{mathrsfs}
\usepackage{amsfonts}
\usepackage{amssymb,bm}
\usepackage{amsmath}
\usepackage{mathtools}
\usepackage{enumitem}
\usepackage{latexsym}
\usepackage{multirow}
\usepackage[pdftex]{color,graphicx}
\RequirePackage[colorlinks,linkcolor=red]{hyperref}
\usepackage{authblk}

\usepackage{hyperref}
\hypersetup{
     colorlinks   = true,
     citecolor    = blue
}
\definecolor{lgray}{gray}{0.70}

\newcommand{\el}{\mathcal{E}}
\newcommand{\E}{\mathbb E}
\newcommand{\prob}{\mathbb P}
\newcommand{\Real}{\mathbb{R}}
\newcommand{\bZ}{\mathbf{Z}}
\newcommand{\bz}{\mathbf{z}}
\newcommand{\gvp}{\varphi}
\newcommand{\Perp}{\perp \! \! \! \perp}

\newcommand{\e}{\epsilon}

\newcommand{\ind}{\text{ind}}

\theoremstyle{definition}
\newtheorem{example}{Example}[section]

\theoremstyle{plain}

\newtheorem{lemma}[example]{Lemma}

\theoremstyle{remark}

\newtheorem{remark}[example]{Remark}
\numberwithin{equation}{section}
\begin{document}
\title{On a Nonparametric Notion of Residual and its Applications}
\author{Rohit Kumar Patra, Bodhisattva Sen\footnote{Supported by NSF CAREER Grant DMS-1150435.}, G\'abor Sz\'ekely}
\affil{Columbia University and National Science Foundation}
\maketitle
\begin{abstract}
\noindent Let $(X, \bZ)$ be a continuous random vector in $\Real  \times \Real^d$, $d \ge 1$. In this paper, we define the notion of a nonparametric residual of $X$ on $\bZ$ that is always independent of the predictor $\bZ$. We study its properties and show that the proposed notion of residual matches with the usual residual (error) in a multivariate normal regression model. Given a random vector $(X, Y, \bZ)$ in $\Real  \times \Real \times \Real^d$, we use this notion of residual to show that the conditional independence between $X$ and $Y$, given $\bZ$, is equivalent to the mutual independence of the residuals (of $X$ on $\bZ$ and $Y$ on $\bZ$) and $\bZ$. This result is used to develop a test for conditional independence. We propose a bootstrap scheme to approximate the critical value of this test. We compare the proposed test, which is easily implementable, with some of the existing procedures through a simulation study.
\end{abstract}

\section{Introduction}
Let $(X, \bZ)$ be a random vector in $\Real \times \Real^d = \Real^{d+1}$, $d \ge 1$. We assume that $(X, \bZ)$ has a joint density on $\Real^{d+1}$. If we want to predict $X$ using $\bZ$ we usually formulate the following regression problem:
\begin{eqnarray}\label{eq:RegMdl}
X = m(\bZ) + \e,
\end{eqnarray}
where $m(\bz) = \E(X|\bZ = \bz)$ is the conditional mean of $X$ given $\bZ = \bz$ and $\e := X - m(\bZ)$ is the {\it residual} (although $\e$ is  usually called the error, and its estimate the residual, for this paper we feel that the term residual is more appropriate). Typically we further assume that the residual $\e$ is {\it independent} of $\bZ$. However, intuitively, we are just trying to break the information in $(X,\bZ)$ into two parts: a part that contains all relevant information about $X$, and the ``residual'' (the left over) which does not have anything to do with the relationship between $X$ and $\bZ$. 

In this paper we address the following question: given any random vector $(X, \bZ)$ how do we define the notion of a ``residual'' of $X$ on $\bZ$ that matches with the above intuition? Thus, formally, we want to find a function $\gvp: \Real^{d+1} \to \Real$ such that the residual $\gvp(X, \bZ)$ satisfies the following two conditions:
\begin{enumerate}
\item[(C.1)] $\;\;\;\;\;$ the residual $\varphi(X, \bZ)$ is independent of the predictor $\bZ$, i.e.,
\begin{eqnarray*}\label{eq:Indep}
\varphi(X, \bZ) \Perp  \bZ, \qquad \mbox{and }
\end{eqnarray*}
\item[(C.2)] $\;\;\;\;\;$ the information content of $(X, \bZ)$ is the same as that of $(  \varphi(X, \bZ), \bZ )$, i.e.,
\begin{equation}\label{eq:Info}
\sigma(X, \bZ) = \sigma(  \varphi(X, \bZ), \bZ ),
\end{equation}
where $\sigma(X, \bZ)$ denotes the $\sigma$-field generated by $X$ and $\bZ$. We can also express~\eqref{eq:Info} as:  there exists a measurable function $h : \Real^{d+1} \to \Real $ such that
\begin{equation}\label{eq:GenX}
X  = h(\bZ, \varphi(X, \bZ));
\end{equation}
see e.g., Theorem 20.1 of~\cite{Bill95}.
\end{enumerate}

In this paper we propose a notion of a residual that satisfies the above two conditions, under any joint distribution of $X$ and $\bZ$. We investigate the properties of this notion of residual in Section~\ref{sec:NPResid}. We show that this notion indeed reduces to the usual residual (error) in the multivariate normal regression model. Further, we use this notion of residual to develop a test for conditional independence.

Suppose now that $(X,Y,\bZ)$ has a joint density   on $\Real \times \Real \times \Real^d = \Real^{d+2}$. The assumption of conditional independence means that $X$ is independent of $Y$ given $\bZ$, i.e., $X \Perp Y |\bZ$. Conditional independence is an important concept in modeling causal relations (\cite{dawid79}, \cite{Pearl00}), in graphical models (\cite{Lauritzen96}; \cite{koller09}), in economic theory (see \cite{Chiappori00}), and in the literature of program evaluations (see \cite{Heckman97}) among other fields. Traditional methods for testing conditional independence are either restricted to the discrete case (\cite{Lauritzen96}; \cite{Agresti02}) or impose simplifying assumption when the random variables are continuous (\cite{Lawrance76}). However, recently there has been a few nonparametric testing procedures proposed for testing conditional independence without assuming a functional form between the distributions of $X,Y$, and $\bZ$. \cite{SuWhite07} consider testing conditional independence based on the difference between the conditional characteristic functions, while \cite{SuWhite08} use the Hellinger distance between conditional densities of $ X$ given $Y$ and $\bZ$, and $X$ given $Y$ to test for conditional independence. A test based on estimation of  the maximal nonlinear conditional correlation is proposed in \cite{Huang10}. \cite{B11} develops  a test based on partial copula.  \cite{KerCondInd07} propose  a measure of conditional dependence of random variables, based on normalized cross-covariance operators on reproducing kernel Hilbert spaces; \cite{Z12} propose another kernel-based conditional independence test. \cite{poczos12} extend the concept of distance correlation (developed by \cite{SzekelyRizzoBakirov07}  to measure dependence between  two random variables or vectors) to characterize conditional dependence.  \cite{SR14} investigate a method that is easy to compute and can capture non-linear dependencies but does not completely characterize conditional independence; also see~\cite{GW12} and the references therein.

In Section~\ref{sec:TestCondInd} we use the notion of residual  defined in Section~\ref{sec:NPResid} to show that the conditional independence between $X$ and $Y$ given $\bZ$ is equivalent to the mutual independence of three random vectors: the residuals of $X$ on $\bZ$ and $Y$ on $\bZ$, and $\bZ$. We reduce this testing of mutual independence to a one sample multivariate  goodness-of-fit test. We further propose a modification of the easy-to-implement \textit{energy} statistic based method (\cite{SzekelyRizzo05}; also see \cite{SzekelyRizzo13}) to test the goodness-of-fit; see Section~\ref{sec:TestMutInd}. In Section~\ref{sec:sub_test_cond} we use our notion of nonparametric residual and the proposed goodness-of-fit test to check the null hypothesis of  conditional independence. Moreover,  we describe a bootstrap scheme to approximate the critical value of this test. 
In Section \ref{sec:simul} we compare the finite sample performance of the procedure proposed in this paper with other available methods in the literature through a finite sample simulation study. We end with a brief discussion, Section~\ref{sec:Disc}, where we point to some open research problems and outline an idea, using the proposed residuals, to define (and test) a nonparametric notion of partial correlation.

\section{A nonparametric notion of residual}\label{sec:NPResid}
Conditions (C.1)--(C.2) do not necessarily lead to a unique choice for $\gvp$. To find a meaningful and unique function $\gvp$ that satisfies conditions (C.1)--(C.2) we impose the following natural restrictions on $\gvp$. We assume that
\begin{enumerate}
\item[(C.3)] $\;\;\;\;\;$ $x \mapsto \gvp(x,\bz)$ is strictly increasing in its support, for every fixed $\bz \in \Real^d$.
\end{enumerate}
Note that condition (C.3) is a slight strengthening of condition (C.2). Suppose that a function $\gvp$ satisfies conditions (C.1) and (C.3). Then any strictly monotone transformation of $\gvp(\cdot, \bz)$ would again satisfy (C.1) and (C.3). Thus,  conditions (C.1) and (C.3) do not uniquely specify $\gvp$. To handle this identifiability issue, we replace condition (C.1) with (C.4), described below.

First observe that, by condition (C.1),  the conditional distribution of the random variable $\gvp(X, \bZ)$ given $\bZ = \bz$  does not depend on $\bz $. We assume that
 \begin{enumerate}
\item[(C.4)] $\;\;\;\;\;$ $\gvp(X, \bZ)| \bZ = \bz$ is uniformly distributed, for all $\bz \in \Real^d$.
\end{enumerate}
Condition (C.4) is again quite natural --  we usually assume that the residual has a fixed distribution, e.g., in regression we assume that the (standardized) residual in normally distributed with zero mean and unit variance. Note that condition (C.4) is slightly stronger than (C.1) and will help us uniquely identify $\gvp$. The following result shows that, indeed, under conditions (C.3)--(C.4), a unique $\gvp$ exists and gives its form.

\begin{lemma}\label{lem:NPError}
Let $F_{X|\bZ}(\cdot| \bz)$ denote the conditional distribution function of $X|\bZ = \bz$. Under conditions (C.3) and (C.4), we have a unique choice of $\gvp(x, \bz)$, given by
\begin{eqnarray*}
\gvp(x, \bz) = F_{X|\bZ}(x| \bz).
\end{eqnarray*}
Also, $h(\bz, u)$ can be taken as
\begin{eqnarray}\label{eq:InvCondDist}
h(\bz, u) =F^{-1}_{X|\bZ}(u|\bz).
\end{eqnarray}
\end{lemma}
\begin{proof} Fix $\bz$ in the support of $\bZ$. Let $u \in (0, 1)$. Let us write $\gvp_\bz(x) =  \gvp(x, \bz)$. By condition (C.4), we have  $\prob[ \gvp(X, \bZ) \le u  | \bZ = \bz  ]  = u$. On the other hand, by (C.3),  $$\prob[ \gvp(X, \bZ) \le u  | \bZ = \bz  ]   =  \prob[ X \le \gvp_\bz^{-1}(u) | \bZ = \bz ] =  F_{X|\bZ}( \gvp_\bz^{-1}(u) | \bz)  .$$  Thus, we have
$$   F_{X|\bZ}( \gvp_\bz^{-1}(u) | \bz) = u, \ \  \mbox{ for all } u \in (0,1),  $$
which is equivalent to $ \gvp_\bz(x) = F_{X|\bZ}(x| \bz)$.

Let $h$ be as defined in~\eqref{eq:InvCondDist}. Then,
$$ h(\bz, \gvp(x, \bz))  = F^{-1}_{X|\bZ}( \gvp(x, \bz)  |\bz)  = F^{-1}_{X|\bZ}( F_{X|\bZ}(x| \bz)  |\bz) = x,       $$
as required.
\end{proof}
Thus from the above lemma, we conclude that in the nonparametric setup, if we want to have a notion of a residual satisfying conditions (C.3)--(C.4) then the residual has to be  $F_{X|\bZ}(X| \bZ)$. The following remarks are in order now.
\begin{remark}
Let us first consider the example when $(X, \bZ)$ follows a multivariate Gaussian distribution, i.e.,
$$ \begin{pmatrix} X \\ \bZ\end{pmatrix} \sim N \left (  \begin{pmatrix} \mu_1 \\ \bm{\mu}_2 \end{pmatrix}, \Sigma :=  \begin{pmatrix} \sigma_{11}& \bm{\sigma}_{12}^\top \\ \bm{\sigma}_{12} & \Sigma_{22}  \end{pmatrix}  \right),  $$ where $\mu_1 \in \Real$, $\mu_2 \in \Real^d$, $\Sigma$ is a $(d+1) \times (d+1)$ positive definite matrix with $\sigma_{11} > 0$, $\sigma_{12} \in \Real^{d \times 1}$ and $\Sigma_{22} \in \Real^{d \times d}$.

Then the conditional distribution of $X$ given $\bZ = \bz$ is  $N(\mu_1 + \bm{\sigma}_{12}^\top \Sigma_{22}^{-1} (\bz - \bm{\mu}_2), \sigma_{11}  - \bm{\sigma}_{12}^\top \Sigma_{22}^{-1} \bm{\sigma}_{12}    )$.
Therefore, we have the following representation in the form of~\eqref{eq:RegMdl}:
$$  X =  \mu_1 + \bm{\sigma}_{12}^\top \Sigma_{22}^{-1} (\bZ - \bm{\mu}_2) +  \Big( X - \mu_1 - \bm{\sigma}_{12}^\top \Sigma_{22}^{-1} (\bZ - \bm{\mu}_2) \Big)  $$
where the usual residual is $X - \mu_1 - \bm{\sigma}_{12}^\top \Sigma_{22}^{-1} (\bZ - \bm{\mu}_2)$, which is known to be independent of $\bZ$. In this case, using Lemma~\ref{lem:NPError}, we get
$$ \gvp(X, \bZ) =  \Phi \left(\frac{ X - \mu_1 - \bm{\sigma}_{12}^\top \Sigma_{22}^{-1} (\bZ - \bm{\mu}_2) }{\sqrt{\sigma_{11}  - \bm{\sigma}_{12}^\top \Sigma_{22}^{-1} \bm{\sigma}_{12}} } \right),$$
where $\Phi(\cdot)$ is the distribution function of the standard normal distribution. Thus  $\gvp(X, \bZ)$ is just a fixed strictly increasing transformation of the usual residual, and the two notions of residual essentially coincide. \\
\end{remark}


\begin{remark}
The above notion of residual does not extend so easily to the case of discrete random variables. Conditions (C.1) and (C.2) are equivalent to the fact that $\sigma(X, \bZ)$ factorizes into two sub $\sigma$-fields as $\sigma(X, \bZ) = \sigma(  \varphi(X, \bZ) ) \otimes \sigma(\bZ )$.  This may not be always possible as can be seen from the following simple example.

Let $(X, Z)$ take values in $\{0, 1\}^2$ such that $\prob[X = i, Z =j] >0$ for all $i, j \in \{0, 1\}$.  Then it can be shown that such a factorization exists if and only if $X$ and  $Z$ are independent, in which case $\gvp(X, Z) = X$. \\
\end{remark}

\begin{remark}
Lemma~\ref{lem:NPError} also gives an way to generate $X$, using $\bZ$ and the residual. We can first generate $\bZ$, following its marginal distribution, and an independent random variable $U \sim \mathcal{U}(0,1)$ (here  $\mathcal{U} (0,1)$ denotes the Uniform distribution on $(0,1)$) which will act as the residual. Then~\eqref{eq:GenX}, where $h$ is defined in~\eqref{eq:InvCondDist}, shows that we can generate $X = F^{-1}_{X|\bZ}(U|\bZ)$. \\
\end{remark}

In practice, we need to estimate the residual $F_{X|\bZ}(X|\bZ)$ from observed data, which can be done both parametrically and non-parametrically. If we have a parametric model for $F_{X|\bZ}(\cdot|\cdot)$, we can estimate the parameters, using e.g., maximum likelihood, etc. If we do not want to assume any structure on $F_{X|\bZ}(\cdot|\cdot)$, we can use any nonparametric smoothing method, e.g., standard kernel methods, for estimation; see~\cite{B11} for such an implementation. We will discuss the estimation of the residuals in more detail in Section~\ref{sec:NPEst}.

\section{Conditional independence}\label{sec:TestCondInd}
Suppose now that $(X,Y,\bZ)$ has a joint density on $\Real \times \Real \times \Real^d = \Real^{d+2}$.
In this section we state a simple result that reduces testing for the conditional independence hypothesis $H_0: X \Perp Y |\bZ$ to a problem of testing mutual independence between three random variables/vectors that involve our notion of residual. We also briefly describe a procedure to test the mutual independence of the three random variables/vectors (see Section~\ref{sec:TestMutInd}). We start with the statement of the crucial lemma.
\begin{lemma}\label{lem:CondInd}
Suppose that $(X,Y,\bZ)$ has a continuous joint density on $\mathbb{R}^{d+2}$. Then, $X \Perp Y |\bZ$ if and only if $F_{X|\bZ}(X|\bZ), F_{Y|\bZ}(Y|\bZ)$ and $\bZ$ are mutually independent.
\end{lemma}
\begin{proof}
Let us make the following change of variable $$ (X,Y,\bZ) \mapsto (U,V,\bZ) := (F_{X|\bZ}(X), F_{Y|\bZ}(Y), \bZ).$$ The joint density of $(U,V,\bZ)$ can be expressed as
\begin{equation}\label{eq:trans}
f_{(U,V,\bZ)}(u,v,\bz) = \frac{f(x,y,\bz)}{f_{X|\bZ=\bz}(x) f_{Y|\bZ=\bz}(y)}  =  \frac{f_{(X, Y)|\bZ=\bz}(x, y)f_\bZ(\bz)}{f_{X|\bZ=\bz}(x) f_{Y|\bZ=\bz}(y)},
\end{equation}
where $x = F_{X|\bZ=\bz}^{-1}(u)$, and $y = F_{Y|\bZ=\bz}^{-1}(v)$. Note that as the Jacobian matrix is upper-triangular, the determinant is the product of the diagonal entries of the matrix, namely, $f_{X|\bZ = \bz}(x)$, $f_{Y|\bZ=\bz}(y)$ and $1$.

If $X \Perp Y |\bZ$ then $f_{(U,V,\bZ)}(u,v,\bz)$ reduces to just $f_\bZ(\bz)$, for $u, v \in (0,1)$, from the definition of conditional independence, which shows that $U,V,\bZ$ are independent (note that it is easy to show that $U,V$ are marginally $\mathcal{U}(0,1)$, the Uniform distribution on $(0,1)$). Now, given that $U,V,\bZ$ are independent, we know that $f_{(U,V,\bZ)}(u,v,\bz) = f_\bZ(\bz)$ for $u, v \in (0,1)$, which from (\ref{eq:trans}) easily shows that $X \Perp Y |\bZ$.
\end{proof}
$\vspace{0.000in}$

\begin{remark}\label{rem:berg}
Note that the joint distribution of $F_{X|\bZ}(X|\bZ)$ and $F_{Y|\bZ}(Y|\bZ)$ is known as the \textit{partial copula}; see e.g.,~\cite{B11}.~\cite{B11} developed a test for conditional independence by testing mutual independence between $F_{X|\bZ}(X|\bZ)$ and  $F_{Y|\bZ}(Y|\bZ)$. However, as the following example illustrates, the independence of $F_{X|\bZ}(X|\bZ)$ and  $F_{Y|\bZ}(Y|\bZ)$ is not enough to guarantee that $X \Perp Y |\bZ$. Let  $W_1, W_2, W_3$ be i.i.d.~$\mathcal{U}(0,1)$ random variables.  Let $X = W_1+W_3$, $Y =W_2$ and $Z = \mathrm{mod}(W_1 + W_2, 1)$, where `$\mathrm{mod}$' stands for the modulo (sometimes called modulus) operation that finds the remainder of the division $W_1 + W_2$ by 1. Clearly, the random vector $(X, Y, Z)$  has a smooth continuous density on $[0,1]^3$. Note that $Z$ is independent of $W_i$, for $i = 1,2$.  Hence, $X, Y$ and $Z$  are pairwise independent.  Thus, $F_{X|\bZ}(X) = F_X(X)$ and $F_{Y|\bZ}(X) = F_Y(Y)$, where $F_X$ and $F_Y$ are the marginal distribution functions of $X$ and $Y$, respectively. From the independence of $X$ and $Y$, $F_X(X)$ and $F_Y(Y)$ are independent. On the other hand, the value of $W_1$ is clearly determined by $Y$ and $Z$, i.e., $W_1 = Z-Y$ if $Y \le Z$ and $W_1 = Z-Y+1$ if $Y>Z$. Consequently,  $X$ and $Y$ are not conditionally independent given $Z$. To see this, note that for every $z \in (0,1)$, $$\E[ X| Y, Z=z ] =     \left\{ \begin{array}{ll}
         z-Y + 0.5 & \mbox{if $Y \le z$}\\
         z - Y +1 + 0.5& \mbox{if $Y > z$,}\end{array} \right.$$ which obviously depends on $Y$. In Remark~\ref{Bergsma2} we illustrate this behavior with a finite sample simulation study. \\
\end{remark}

\begin{remark} We can extend the above result to the case when $X$ and $Y$ are random vectors in $\Real^p$ and $\Real^q$, respectively. In that case we define the conditional multivariate distribution transform $F_{X|\bZ}$ by successively conditioning on the co-ordinate random variables, i.e., if $X = (X_1,X_2)$ then we can define $F_{X|\bZ}$ as $(F_{X_2|X_1,\bZ},  F_{X_1|\bZ})$. With this definition, Lemma~\ref{lem:CondInd} still holds.  \\
\end{remark}

To use Lemma~\ref{lem:CondInd} to test the conditional independence between $X$ and $Y$ given $\bZ$, we need to first estimate the residuals $F_{X|\bZ}(X|\bZ)$ and  $F_{Y|\bZ}(Y|\bZ)$ from observed data, which can be done by any nonparametric smoothing procedure, e.g., standard kernel methods (see Section~\ref{sec:NPEst}). Then, any procedure for testing the mutual independence of $F_{X|\bZ}(X|\bZ), F_{Y|\bZ}(Y|\bZ)$ and $\bZ$ can be used. In this paper we advocate the use of the {\it energy} statistic (see \cite{RizzoSzekely10}), described briefly in the next subsection, to test the mutual independence of three or more random variables/vectors.

\subsection{Testing mutual independence of three or more random vectors with known marginals}\label{sec:TestMutInd}
Testing independence of two random variables (or vectors) has received much recent attention in the statistical literature; see e.g.,~\cite{SzekelyRizzoBakirov07}, \cite{KerIndepALT05}, and the references therein. However, testing the mutual independence of three or more random variables is more complicated and we could not find any easily implementable method in the statistical literature.

In this sub-section, we test the mutual independence of three or more random variables (vectors) with known marginals by converting the problem to a one-sample goodness-of-fit test for multivariate normality. In the following we briefly describe our procedure in the general setup.

Suppose that we have $r \ge 3$ continuous random variables (or vectors) $V_1, \ldots, V_r$ and we want to test their mutual independence. We assume that we know the marginal distributions of $V_1, \ldots, V_r$; without loss of generality, we can assume that $V_i$'s are standard Gaussian random variables (vectors). We write $T:= (V_1, V_2, \ldots, V_r) \in \Real^k$ and introduce $T_{\ind} := (V_1^*, V_2^*, \ldots, V_r^*)$ where $V_j^*$ is an i.i.d.~copy of $V_j$, $j=1,2, \ldots, r$, but in $T_{\ind}$ the coordinates, $V_1^*, V_2^*, \ldots, V_r^*$, are independent. To test the mutual independence of $V_1, V_2, \ldots, V_r$ all we need to test now is whether $T$ and $T_{\ind}$ are identically distributed. If we observed a sample from $T$, we can test for the equality of distributions of $T$ and $T_{\ind}$ through a one-sample goodness-of-fit test for the standard multivariate normal distribution, i.e., $$H_0: T \sim N(\textbf{0},\textbf{I}_{k\times k}),$$ as $T_{\ind}\sim N(\textbf{0},\textbf{I}_{k\times k})$, where $\textbf{I}_{k \times k}$ is the identity matrix of order $k$ and  $\textbf{0} := (0, \ldots, 0) \in \Real^{k}.$


In this paper we consider the following {\it energy} statistic (see~\cite{SzekelyRizzo05} and \cite{RizzoSzekely10})
\begin{equation}\label{eq:EStat}
\Lambda(T) = 2 \E \|T - T_{\ind}\| - \E \|T - T'\| - \E \|T_{\ind} - T_{\ind}'\|,
\end{equation}
where $T'$ and $T_{\ind}'$ are i.i.d.~copies of $T$ and $T_{\ind}$, respectively ($\|\cdot\|$ denotes the Euclidean norm). Note that $\Lambda(T)$ is always nonnegative, and equals 0, if and only if  $T$ and $T_{\ind}$ are identically distributed, i.e., if and only if $V_1, V_2, \ldots, V_r$ are mutually independent (see Corollary 1 of~\cite{SzekelyRizzo05}).

  Suppose now that we observe $n$  i.i.d.~samples $T_1, \ldots, T_n$ of $T$. The (scaled)  sample version of the energy statistic for testing the goodness-of-fit hypothesis is
  \begin{equation}\label{eq:teststat}
\mathcal{E}_n(T_1,\ldots, T_n) :=2 \sum_{i=1}^n  \E \|T_i-T_\ind\|  - \frac{1}{n} \sum_{i=1}^n\sum_{j=1}^{n} \|T_i-T_j\|- n \E \|T_\ind-T^\prime_\ind\|.
  \end{equation}
Note that the first expectation in the above display is with respect to $T_\ind$. Under the null hypothesis of mutual independence, the test statistic $\mathcal{E}_n(T_1,\ldots, T_n)$ has a limiting distribution, as $n \rightarrow \infty,$ while under the alternative hypothesis $\mathcal{E}_n(T_1,\ldots, T_n)$ tends to infinity; see Section 4 of \cite{SzekelyRizzo05} and Section 8 of \cite{SzekelyRizzo13} for detailed discussions. Thus any test that rejects the null for large values of  $\mathcal{E}_n(T_1,\ldots, T_n)$ is consistent against general alternatives.

 As   $T_\ind$ and $T^\prime_\ind$ are  i.i.d.~$N(\textbf{0}, \textbf{I}_{k\times k})$ random variables.  The statistic $\mathcal{E}_n(T_1,\ldots, T_n)$ is easy to compute:
  $$\E\|T_\ind-T_\ind^\prime\| =\sqrt{2}\E \|T_{ind}\|= 2 \frac{\Gamma \big(\frac{d+3}{2}\big)}{\Gamma \big( \frac{d+2}{2}\big)}$$ and for any $a\in \Real^{d+2}$, we have
  $$\E\|a-T_\ind\| =\frac{\sqrt{2}\Gamma \big(\frac{d+3}{2}\big)}{\Gamma \big( \frac{d+2}{2}\big)} + \sqrt{\frac{2}{\pi}} \sum_{k=0}^\infty \frac{(-1)^k}{k!\, 2^k} \frac{|a|^{2k+2}}{(2k+1)(2k+2)} \frac{\Gamma \big( \frac{d+3}{2}\big)\Gamma \big( k+\frac{3}{2}\big)}{\Gamma \big( k+\frac{d}{2}+2\big)}.$$

 The expression for $\E\|a-T_\ind\|$ follows from the discussion in \cite{Zacks81} (see page 55). See the
source code ``energy.c'' in  the \textit{energy} package  of R language (\cite{Rlang}) for a fast implementation of this; also see \cite{SzekelyRizzo13}.

\subsection{Testing conditional independence} \label{sec:sub_test_cond}
In this sub-section we use Lemma \ref{lem:CondInd} and the test for mutual independence proposed in the  previous sub-section (Section~\ref{sec:TestMutInd}) to test for the conditional independence of $X$ and $Y$ given $\bZ.$ We start with a simple lemma.


\begin{lemma} \label{lem:CondIndeqiv}
Suppose that $(X,Y,\bZ)$ has a continuous joint density on $\mathbb{R}^{d+2}$. Then $X \Perp Y |\bZ$ if and only if $$W:=(F_{X|\bZ}(X|\bZ), F_{Y|\bZ}(Y|\bZ), F_\bZ(\bZ)) \sim \mathcal{U}([0,1]^{d+2}),$$ where $F_\bZ(\bz) = \left(F_{Z_d|Z_{d-1},\ldots, Z_1}(z_d|z_{d-1},\ldots, z_1), \ldots, F_{Z_2|Z_1}(z_2|z_1), F_{Z_1}(z_1)\right),$ $\bZ =$ \\$ (Z_1,\ldots, Z_d),$ $\textbf{z}=(z_1,\ldots, z_d),$ and $\mathcal{U}([0,1]^{d+2})$ denote the Uniform distribution on $[0,1]^{d+2}$.
\end{lemma}
\begin{proof}
Note that by Lemma~\ref{lem:CondInd},   $X \Perp Y |\bZ$ if and only if $F_{X|\bZ}(X|\bZ),$ $F_{Y|\bZ}(Y|\bZ)$ and $\bZ$ are mutually independent. Furthermore, note that  $F_{X|\bZ}(X|\bZ),$ $F_{Y|\bZ}(Y|\bZ)$ are i.i.d.~$\mathcal{U}(0,1)$ random variables. Thus the proof of the lemma will be complete if we  show that $F_\bZ(\bZ) \sim \mathcal{U}([0,1]^d)$.

As  each of $F_{Z_d|Z_{d-1},\ldots, Z_1}(Z_d|Z_{d-1},\ldots, Z_1), \ldots, F_{Z_2|Z_1}(Z_2|Z_1),$ and  $F_{Z_1}(Z_1)$ are $\mathcal{U}(0,1)$ random variables, it is enough to show that they are mutually independent. For simplicity of notation, we will only prove the independence of $F_{Z_2|Z_1}(Z_2|Z_1)$ and  $F_{Z_1}(Z_1)$, independence of other terms can be proved similarly. Note that
\begin{align*}
\prob(F_{Z_2|Z_1}(Z_2|Z_1) \le z_2 | F_{Z_1}(Z_1)=z_1) ={}& \prob(F_{Z_2|Z_1}(Z_2|Z_1) \le z_2 | Z_1=F_{Z_1}^{-1}(z_1))\\
={}&\prob\Big(Z_2 \le F_{Z_2|Z_1}^{-1}\big(z_2| F_{Z_1}^{-1}(z_1)\big) \Big| Z_1=F_{Z_1}^{-1}(z_1)\Big)\\
={}&F_{Z_2|Z_1} \Big(F_{Z_2|Z_1}^{-1}\big(z_2| F_{Z_1}^{-1}(z_1)\big) |F_{Z_1}^{-1}(z_1)\Big)\\
={}&z_2.
\end{align*}
As the conditional distribution of $F_{Z_2|Z_1}(Z_2|Z_1)$ given $ F_{Z_1}(Z_1) = z_1$ does not depend on $z_1$, we have  that  $F_{Z_2|Z_1}(Z_2|Z_1)$ and  $F_{Z_1}(Z_1)$ are independent.
\end{proof}

Let us now assume  $X \Perp Y |\bZ$ and  define
\begin{equation*} \label{eq:T_def}
W:=\left(F_{X|\bZ}(X|\bZ), F_{Y|\bZ}(Y|\bZ), F_{Z_d|\bZ_{-d}}(Z_d|\bZ_{-d}), \ldots, F_{Z_2|Z_1}(Z_2|Z_1), F_{Z_1}(Z_1)\right).
\end{equation*}
By Lemma~\ref{lem:CondIndeqiv}, we have
\begin{equation*} \label{eq:eq_dist}
W\stackrel{\mathcal D}{=} (U_1, \dots, U_{d+2}),
\end{equation*}
where $U_1, U_2, \ldots,  U_{d+2}$ are i.i.d.~$\mathcal{U}(0,1)$ random variables.  An equivalent formulation is
\begin{equation} \label{eq:mvn}
H_0: T:= \Phi^{-1} (W) \stackrel{\mathcal D}{=} N(\textbf{0}, \textbf{I}_{(d+2) \times  (d+2)}),
\end{equation}
 where $\Phi$ is the distribution function corresponding to the standard Gaussian random variable, and for any $\textbf{a} \in \Real^{d+2}$, $\Phi^{-1} (\textbf{a}) := (\Phi^{-1}(a_1), \ldots, \Phi^{-1}(a_{d+2})).$

 We observe i.i.d.~data $\{(X_i,Y_i,\bZ_i): i = 1,\ldots, n\}$ from the joint distribution of $(X,Y,\bZ)$ and we are interested in testing $X \Perp Y |\bZ$. Suppose first that the distribution functions $F_{X| \bZ}(\cdot|\cdot),  F_{Y| \bZ}(\cdot|\cdot),$ and  $F_{\bZ}(\cdot)$ are known.  Then we have an i.i.d.~sample $T_1,\ldots, T_n$ from $T$, where
\begin{equation} \label{eq:data_ver}
 T_i:=\Phi^{-1}(F_{X|\bZ}(X_i|\bZ_i), F_{Y|\bZ}(Y_i|\bZ_i), F_{\bZ}(\bZ_i)).
\end{equation}
Now we can use the the test statistic \eqref{eq:teststat} to  test the hypothesis of conditional independence.

 As the true conditional distribution functions $F_{X| \bZ},  F_{Y| \bZ},$ and  $F_{\bZ}$ are unknown, we can replace them by their estimates $\widehat F_{X|\bZ}, \widehat F_{Y|\bZ},$  and $\widehat F_{\bZ}$, respectively, where $\widehat F_\bZ (\bz) =\left( \widehat F_{Z_d|Z_{d-1},\ldots, Z_1}(z_d|z_{d-1},\ldots, z_1), \ldots,\widehat F_{Z_2|Z_1}(z_2|z_1), \widehat F_{Z_1}(z_1)\right)$; see Section \ref{sec:NPEst} for more details on how to compute these estimates. Let us now define
 \begin{equation} \label{eq:data_hat_ver}
 \widehat T_i:=\Phi^{-1}(\widehat F_{X|\bZ}(X_i|\bZ_i), \widehat F_{Y|\bZ}(Y_i|\bZ_i), \widehat F_{\bZ}(\bZ_i)),
 \end{equation}
 for $i = 1, 2,\ldots, n.$  We will use
 \begin{equation} \label{eq:en_hat}
\widehat{\mathcal{E}_n}:=  \mathcal{E}_n(\hat{T}_1, \ldots \hat{T}_n)
    \end{equation} to test the hypothesis of conditional independence.
\subsubsection{Approximating the asymptotic distribution through bootstrap}

The limiting behavior of $\el_n$ is not very useful in computing the critical value of the test statistic $\widehat{\mathcal{E}_n}$ proposed in the the previous sub-section. In a related but slightly different problem studied in~\cite{sen14}, it was shown that, the analogous versions of $\el_n$ and $\widehat{\mathcal{E}_n}$  have very different limiting distributions.

In independence testing problems it is quite standard and natural to approximate the critical value of the test, under $H_0$, by using a permutation test; see e.g.,~\cite{SzekelyRizzo09}, \cite{gretton07}. However, in our problem as we use $\hat{T}_i$ instead of $T_i$, the permutation test is not valid; see~\cite{sen14}.

In this sub-section, we propose a bootstrap procedure to approximate the distribution of $\widehat{\mathcal{E}_n}$, under the null hypothesis of conditional independence.  We now describe the bootstrap procedure.
%
%
%
%
%
%
%
%
%
%
Let $\mathbb{P}_{n,\bZ}$ be the empirical distribution of $\bZ_1, \ldots,\bZ_n$.
\begin{enumerate}[label=\bfseries Step \arabic*:]

\item 	 Generate an i.i.d.~sample $\{U_{i,1}^*, U_{i,2}^*, \bZ^*_{n,i}\}_{ 1 \le i \le n}$ of size $n$ from the measure $\mathcal{U}(0,1) \times \mathcal{U}(0,1) \times \mathbb{P}_{n,\bZ}$; recall that $\mathcal{U}(0,1)$ denotes the Uniform distribution on $(0,1).$

\item 	The bootstrap sample is then $\{X^*_{n,1}, Y^*_{n,1}, \bZ^*_{n,1}\}_{ 1 \le i \le n},$ where
\begin{equation}
X^*_{n,i} := \widehat{F}^{-1}_{X|Z}(U_{i,1}^*|\bZ_{n,1}^*) \qquad \text{and} \qquad Y^*_{n,i} := \widehat{F}^{-1}_{Y|Z}(U_{i,2}^*|\bZ_{n,1}^*).
\end{equation}

\item Use the bootstrap sample $\{X^*_{n,i}, Y^*_{n,i}, \bZ^*_{n,i}\}_{ 1 \le i \le n}$ to get smooth estimators $\widehat F^*_{X|\bZ}, \widehat F^*_{Y|\bZ},$  and $\widehat F^*_{\bZ}$ of $F_{X| \bZ},  F_{Y| \bZ},$ and  $F_{\bZ}$; see Section \ref{sec:NPEst} for a discussion on smooth estimation of the conditional distribution functions.

\item	Compute the bootstrap test statistic  $\mathcal{E}^*_n:= \mathcal{E}_n(\widehat{T}^*_1, \ldots, \widehat{T}^*_n) $ where
{\small \begin{equation}
\widehat{T}^*_i= \Phi^{-1} \big(\widehat F^*_{X|\bZ}(X^*_{n,i}|\bZ_{n,i}^*), \widehat F^*_{Y|\bZ}(Y^*_{n,i}|\bZ^*_{n,i}),  \widehat F^*_{\bZ}(\bZ^*_{n,i})).
\end{equation} }
\end{enumerate}
We can now approximate the distribution of $\widehat{\mathcal{E}_n}$ by the conditional distribution of $\mathcal{E}_n^*$ given the data $\{X_i, Y_i,\bZ_i\}_{ 1 \le i \le n}.$
In Section \ref{sec:simul} we study  the finite sample performance of the above procedure through a  simulation study and illustrate that our procedure indeed yields a valid test for conditional independence.

\begin{remark}
	In steps 1 and 2 above, we generate the bootstrap sample from the approximated joint distribution of $(X, Y, \bZ)$ under the null hypothesis of conditional independence. In steps 3 and 4 we mimic the evaluation of the test statistic $\widehat{\mathcal{E}_n}$ using the bootstrap sample. This is an example of a model based bootstrap procedure.~\cite{sen14} prove the consistency of a similar bootstrap procedure in a related problem.  As the sample size increases the approximated joint distribution of $(X, Y, \bZ)$ (under $H_0$) would converge to the truth and the bootstrap distribution would replicate the distribution of $\widehat{\mathcal{E}_n}$.
\end{remark}

\subsection{Nonparametric estimation of the residuals}\label{sec:NPEst}
In this sub-section we discuss procedures to nonparametrically estimate  $ F_{X| \bZ},  F_{Y| \bZ},$ and  $F_{\bZ}$ given data $\{X_i, Y_i,\bZ_i\}_{ 1 \le i \le n}.$ The nonparametric estimation of the conditional distribution functions would involve smoothing. In the following we briefly describe the standard approach to estimating the conditional distribution functions using  kernel smoothing techniques (also see~\cite{LeeLeePark06}, \cite{YuJones98}, and \cite{HallWolffYao99}). For notational simplicity, we restrict to the case $d=1$, i.e., $\bZ$ is a real-valued random variable. Given an i.i.d.~sample of $\{(X_i,Z_i): i = 1,\ldots, n\}$ from $f_{X,Z}$, the joint density of $(X,Z)$, we can use the following kernel density estimator of $f_{X,Z}$: $$ \widehat f_n(x,z) =  \frac{1}{n h_{1,n} h_{2,n}} \sum_{i=1}^n k \left( \frac{x - X_i}{h_{1,n}} \right) k \left( \frac{z - Z_i}{h_{2,n}} \right)$$ where $k$ is a symmetric probability density on $\mathbb{R}$ (e.g., the standard normal density function), and $h_{i,n}, i=1,2$, are the smoothing bandwidths. It can be shown that if $n h_{1,n} h_{2,n} \rightarrow \infty$ and $\max\{h_{1,n}, h_{2,n}\} \rightarrow 0$, as $n \rightarrow \infty,$ then $\widehat f_n(x,z) \stackrel{P}{\rightarrow} f_{X,Z}(x,z)$. In fact, the theoretical properties of the above kernel density estimator are very well studied; see e.g., \cite{FG96} and \cite{EM05} and the references therein. For the convenience of notation, we will write $h_{i,n}$ as $h_i$, $i=1,2$.

The conditional density of $X$ given $Z$ can then be estimated by $$\widehat f_{X|Z}(x|z) = \frac{\widehat f_n(x,z)}{\widehat f_Z(z)} = \frac{\frac{1}{n h_{1} h_{2}} \sum_{i=1}^n k \left( \frac{x - X_i}{h_{1}} \right) k \left( \frac{z - Z_i}{h_{2}} \right)}{\frac{1}{n h_{2}} \sum_{i=1}^n k \left( \frac{z - Z_i}{h_{2}} \right)}.$$
Thus the conditional distribution function of $X$ given $Z$ can be estimated as $$ \widehat F_{X|Z}(x|z) = \frac{\int_{-\infty}^ x \widehat f_n(t,z) \; dt}{\widehat f_Z(z)} = \frac{\frac{1}{n h_{2}} \sum_{i=1}^n K \left( \frac{x - X_i}{h_{1}} \right) k \left( \frac{z - Z_i}{h_{2}} \right)}{\frac{1}{n h_{2}} \sum_{i=1}^n k \left( \frac{z - Z_i}{h_{2}} \right)} = \sum_{i=1}^n  w_i(z) K \left( \frac{x - X_i}{h_{1}} \right) $$ where $K$ is the distribution function corresponding to $k$ (i.e., $K(u)  = \int_{-\infty}^u k(v) \; dv$) and $w_i(z) = \frac{\frac{1}{n h_{2}} k \left( \frac{z - Z_i}{h_{2}} \right)}{\frac{1}{n h_{2}} \sum_{j=1}^n k \left( \frac{z - Z_j}{h_{2}} \right)}$ are weights that sum to one for every $z$.  Least square cross-validation method proposed in \cite{hall2004cross} can be used to find the optimal choices for $h_1$ and $h_2.$  For general $d$, the optimal parameters must satisfy $h_1 \sim n^{-2/(d+4)}$ and $h_2 \sim n^{-1/(d+4)};$ see  Section 6.2 of \cite{LiRacine07}  and \cite{lilira13} for a thorough discussion.

\begin{remark}\label{Bergsma2} Now we  provide empirical evidence  for the failure of the test proposed in~\cite{B11} in the example discussed in Remark~\ref{rem:berg}. We plot (see Figure~\ref{fig:berg}) the histogram of $p$-values obtained from the proposed test (see Section~\ref{sec:sub_test_cond}) and that of the $p$-values obtained from testing the independence of  $F_{X|\bZ}(X|\bZ)$ and $F_{Y|\bZ}(Y|\bZ)$ (using their estimates $\widehat F_{X|\bZ}(\cdot|\cdot)$ and $\widehat F_{Y|\bZ}(\cdot|\cdot)$). We use the distance covariance test statistic (see \citet{SzekelyRizzoBakirov07}) to test for the independence of $F_{X|\bZ}(X|\bZ)$ and  $F_{Y|\bZ}(Y|\bZ)$. Figure~\ref{fig:berg} demonstrates that a test for  mutual independence of $F_{X|\bZ}(X|\bZ)$ and  $F_{Y|\bZ}(Y|\bZ)$ can fail to capture the conditional dependence between $X$ and $Y$ given $\bZ$.
\end{remark}
\begin{figure}[h!]
\includegraphics[scale=.8]{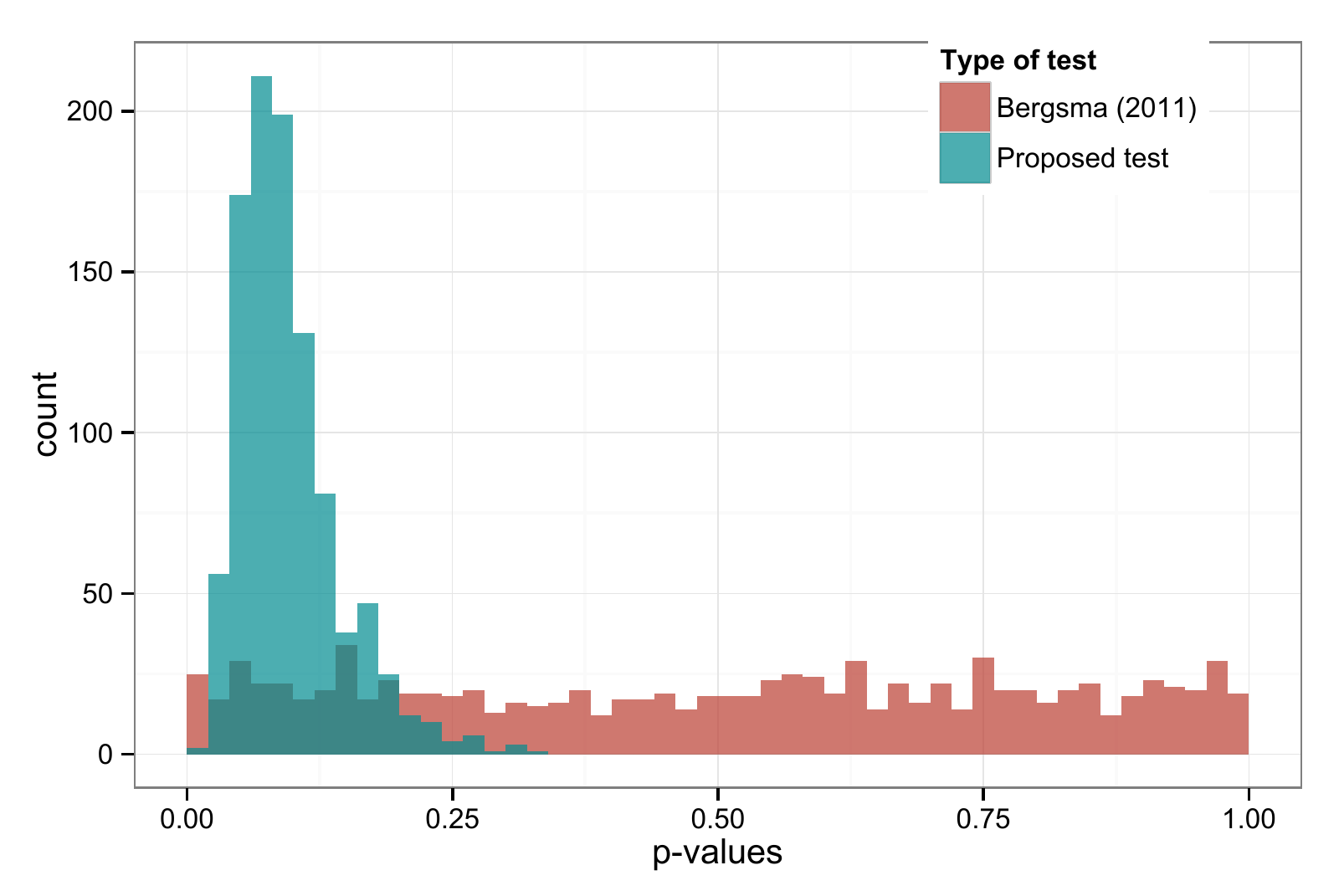}
\caption{Histograms of $p$-values (estimated using 1000 bootstrap samples) over $1000$ independent replications.  Here, for  $i=1,\ldots,200$, $\{X_i,Y_i,Z_i\}$ are i.i.d.~samples from the example discussed in Remark \ref{rem:berg}.}
\label{fig:berg} 
\end{figure}

\section{Simulation}\label{sec:simul}
We now investigate the finite sample performance of the testing procedure developed in this paper through a simulation study. We also compare the performance of the our testing procedure to those  proposed in \cite{KerCondInd07} and \cite{Z12}. We denote the the testing procedure proposed in \cite{KerCondInd07} by $CI_{perm}$ and use $KCI$ to denote the kernel based conditional independence test proposed in \cite{Z12}.

To  illustrate  and compare the performance of different testing procedures, we consider  the following sampling scenario borrowed from \cite{Z12}.  Let us assume that  $X$ and $Y$ are only dependent on $Z_1$ (the first coordinate of $\bZ$) and that all other conditioning variables are independent of $X,Y,$ and $Z_1.$ We assume that  $\bZ \sim N_d(\textbf{0}, \sigma^2_z \textbf{I}_{d\times d})$, $X:= W+ Z_1+ \epsilon,$ and $Y:= W+ Z_1+ \epsilon^\prime,$ where $\epsilon, \epsilon^\prime,$  and $W$ are three independent mean zero Gaussian random variables. Moreover, we assume that $\epsilon, \epsilon^\prime,$  and $W$ are independent of $\bZ,$ $var(\epsilon)=var(\epsilon^\prime)=\sigma^2_E,$ and $var(W)=
\sigma^2_W,$ where for any real random variable $V$,  $var(V)$ denotes its variance.  Note that $X \Perp Y |\bZ$  if and only if $\sigma_W=0.$

  In our finite sample simulations  we fixed $\sigma_E= 0.3 $ and $\sigma_z=0.2$.  We  generate $500$ i.i.d.~samples $\{X_i, Y_i, \bZ_i\}_{1 \le i \le 500}$ for each of $d=1, 3,$ and $5$ and for different values of $\sigma_W.$ For each such sample, we use 1000 bootstrap replicates to estimate the $p$-value of the proposed test procedure. We have used the ``\texttt{np}" (see \cite{np}) package in R  (\cite{R}) to estimate the conditional distribution functions with the tuning parameters chosen using least-squares cross validation (see Section~\ref{sec:NPEst}). In Figure \ref{fig:power_curve} we plot the  power (estimated using 500 independent experiments) of the testing procedure proposed in Section \ref{sec:sub_test_cond} along with those of  $CI_{perm}$ and $KCI$  as $\sigma_W$ increases from $0$ to $0.25$, for dimensions $1, 3,$ and $5$. We fix the significance level at $0.05$.

\begin{figure}[h!]
\includegraphics[width=.65\paperwidth]{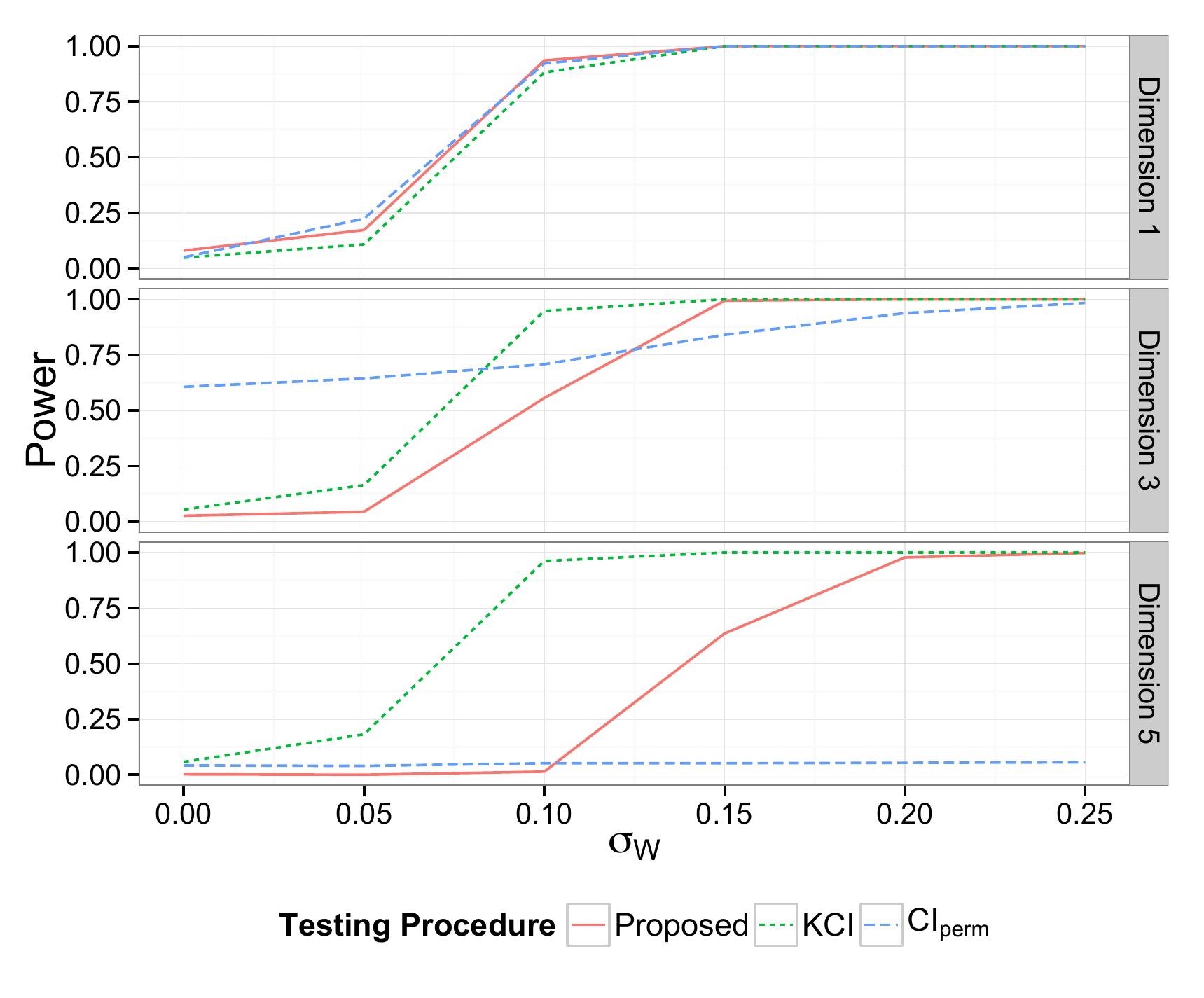}
\caption{The power (at significance level $0.05$) of the three testing procedures for sample size $n=500$ as the dimension $d$ and $\sigma_W$ increase.}
\label{fig:power_curve} 
\end{figure}
%

 The distribution of the $KCI$ test  statistic  under the null hypothesis of conditional independence is estimated with  a Monte Carlo procedure suggested in \cite{Z12}. To implement the $CI_{perm}$ and the $KCI$ testing procedures, we have used the MATLAB source codes provided in \cite{Z12}; the source code can be found at \url{http://people.tuebingen.mpg.de/kzhang/KCI-test.zip}. The R  language codes used to implement our procedure are available at \url{http://stat.columbia.edu/~rohit/research.html}.


Observe that for $CI_{perm}$, the probability of  type I error  is much greater than the significance level for  $d=3$. Furthermore, for $d=5$, it fails to detect the alternative for all values of $\sigma_W$. The performance of $CI_{perm}$ is sensitive to the dimension of the conditioning variable. The probability of type I error for both the proposed and the $KCI$ testing procedures are around the specified significance level. Moreover, the powers of $KCI$ and the proposed test increase to $1$ as $\sigma_W$ increases. Overall, we think that for this simulation scenario the $KCI$ method has the best performance.


\section{Discussion}\label{sec:Disc}
Given a random vector $(X, \bZ)$ in $\Real  \times \Real^d = \Real^{d+1}$ we have defined the notion of a nonparametric residual of $X$ on $\bZ$ as $F_{X|\bZ}(X|\bZ)$, which is always independent of the response $\bZ$. We have studied the properties of the nonparametric residual and showed that it indeed reduces to the usual residual in a multivariate normal regression model. However,  nonparametric estimation of $F_{X|\bZ}(\cdot|\cdot)$ requires smoothing techniques, and hence suffers from the curse of dimensionality. A natural way of mitigating this curse of dimensionality could be to use dimension reduction techniques in estimating the residual $F_{X|\bZ}(X|\bZ)$.  Another alternative would be to use a parametric model for the conditional distribution function.

Suppose now that $(X,Y,\bZ)$ has a joint density on $\Real \times \Real \times \Real^d = \Real^{d+2}$. We have used this notion of residual to show that the conditional independence between $X$ and $Y$, given $\bZ$, is equivalent to the mutual independence of the residuals $F_{X|\bZ}(X|\bZ)$ and $F_{Y|\bZ}(Y|\bZ)$ and the predictor $\bZ$. We have used this result to propose a test for conditional independence, based on the energy statistic. 

We can also use these residuals to come up with a nonparametric notion of partial correlation.  The partial correlation of $X$ and $Y$ measures the degree of association between $X$ and $Y$, removing the effect of $\bZ$. In the nonparametric setting, this reduces to measuring the dependence between the residuals $F_{X|\bZ}(X|\bZ)$ and $F_{Y|\bZ}(Y|\bZ)$. We can use distance covariance (\cite{SzekelyRizzoBakirov07}), or any other measure of dependence, for this purpose. We can also test for zero partial correlation by testing for the independence of the residuals $F_{X|\bZ}(X|\bZ)$ and $F_{Y|\bZ}(Y|\bZ)$. \newline

\noindent {\bf Acknowledgements:} The second author would like to thank Arnab Sen for many helpful discussions, and for his help in writing parts of the paper. He would also like to thank Probal Chaudhuri for motivating the problem. The research of second and third authors is supported by National Science Foundation.
\bibliographystyle{elsarticle-harv}
\bibliography{cond_ind}

\end{document}